\documentclass[10pt,a4paper,envcountsame]{llncs}

\usepackage{comment}
\usepackage{amsmath}
\usepackage{amsfonts}
\usepackage{amssymb}
\usepackage{xspace}
\usepackage{pdfsync}
\usepackage[textwidth=2cm,textsize=small]{todonotes}

\newcommand{\ignore}[1]{}

\newcommand{\C}{\mathcal{C}}
\newcommand{\Cvas}{\C_{\text{\tiny SEC-VAS}}}
\newcommand{\F}{\mathcal{F}}
\newcommand{\G}{\mathcal{G}}

\newcommand{\A}{\mathcal{A}}
\newcommand{\B}{\mathcal{B}}

\newcommand{\N}{\mathbb{N}}
\newcommand{\Z}{\mathbb{Z}}

\newcommand{\powset}[1]{{\cal P}(#1)}

\newcommand{\fin}{\textsc{fin}}

\newcommand{\parim}{\Pi}
\newcommand{\emptyword}{\varepsilon}

\newcommand{\skel}{\textsc{skel}}

\newcommand{\myparagraph}[1]{\vskip 0.3cm \noindent\textbf{#1.}}

\newcommand{\sepsep}{ }

\newcommand{\vas}{VAS\xspace}
\newcommand{\vases}{{\vas}es\xspace}
\newcommand{\vasreach}{\vas reachability\xspace}
\newcommand{\vaslangs}{\vas languages\xspace}
\newcommand{\vass}{VASS\xspace}
\newcommand{\vasses}{{\vass}es\xspace}
\newcommand{\set}[1]{\{#1\}}
\newcommand{\setof}[2]{\set{#1 \mid #2}}

\newcommand{\mycap}{\, \cap \,}

\begin{document}

\title{Regular Separability of Parikh Automata} % Languages}

\author{Lorenzo Clemente\inst{1}%\thanks{lc}
\and Wojciech Czerwi\'nski\inst{1}%\thanks{wc}
\and S{\l}awomir Lasota\inst{1}%\thanks{sl}
\and Charles Paperman\inst{2}}
\institute{University of Warsaw \and University of T{\"u}bingen}

\maketitle

\begin{abstract}
We investigate a subclass of languages recognized by vector addition systems, namely languages of 
nondeterministic Parikh automata. While the regularity problem (is the language of a given automaton regular?) 
is undecidable for this model, we show surprising decidability of the regular separability problem:
given two Parikh automata, is there a regular language that contains one of them and is disjoint from the other?
%
%We show that regular separability problem is decidable for the languages recognized by 
%Parikh automata. We also generalize this result by showing that under some conditions on the class $\C$
%of subsets of $\N^d$ the regular separability problem for $\C$-Parikh automata is decidable.
\end{abstract}

% !TEX root = main.tex

\section{Introduction}
In this paper we investigate separability problems for languages of finite words.
%or for sets of vectors from $\N^d$ or sets of finite words from $\Sigma^*$ called languages.
We say that a language $U$ is \emph{separated from} a language $V$ by $S$ if $U \subseteq S$ and $V \cap S = \emptyset$.
In the sequel we also often say that $U$ \emph{and} $V$ are \emph{separated by} $S$.
For two families of languages $\F$ and $\G$, the \emph{$\F${\sepsep}separability problem for $\G$}
asks, given two given languages $U, V \in \G$, whether $U$ is separated from $V$ by some language from $\F$.
The same notion of separability makes clearly sense if $\F$ and $\G$ are classes of sets of vectors instead of classes of languages.

Concretely, in this paper we mainly consider $\F$ to be regular languages, and $\G$ to be the languages of Parikh automata;
or $\F$ the unary sets, and $\G$ the semilinear sets.

\myparagraph{Motivation}
Separability is a classical problem in theoretical computer science.
It was investigated most extensively in the area of formal languages, for $\G$ being the family of all regular word languages.
Since regular languages are effectively closed under complement, the $\F$ separability problem
is a generalization of the $\F$ characterization problem, which asks whether a given language belongs to $\F$.
Indeed, $L \in \F$ if and only if $L$ is separated from its complement by some language from $\F$.
Separability problems for regular languages attracted recently a lot of attention,
which resulted in establishing the decidability of $\F$ separability for the family $\F$ of separators being
the piecewise testable languages~\cite{DBLP:conf/icalp/CzerwinskiMM13,DBLP:conf/mfcs/PlaceRZ13}
(recently generalized to finite ranked trees \cite{DBLP:conf/icalp/Goubault-Larrecq16}),
the locally and locally threshold testable languages~\cite{DBLP:conf/fsttcs/PlaceRZ13},
the languages definable in first order logic~\cite{DBLP:journals/corr/PlaceZ14},
and the languages of certain higher levels of the first order hierarchy~\cite{DBLP:conf/icalp/PlaceZ14},
among others.

Separability of nonregular languages attracted little attention till now.
The reasons for this may be twofold. First, for regular languages one can use standard algebraic tools, like syntactic monoids,
and indeed most of the results have been obtained with the help of such techniques.
Second, some strong intractability results have been known already since 70's, when Szymanski and Williams proved
that regular{\sepsep}separability of context-free languages is undecidable~\cite{DBLP:journals/siamcomp/SzymanskiW76}.
Later Hunt~\cite{DBLP:journals/jacm/Hunt82a} generalized this result: he showed that $\F$-separability of context-free languages
is undecidable for every class $\F$ which is closed under finite boolean combinations
and contains all languages of the form $w\Sigma^*$ for $w \in \Sigma^*$. This is a very weak condition, so it seemed
that nothing nontrivial can be done outside regular languages with respect to separability problems.
Furthermore, Kopczy\'{n}ski has recently shown that regular{\sepsep}separability is undecidable
even for languages of visibly pushdown automata~\cite{DBLP:journals/corr/Kopczynski15a},
thus strengthening the result by Szymanski and Williams.
On the positive side, piecewise testable{\sepsep}separability has been shown decidable
for context-free languages, languages of vector addition systems (\vaslangs),
and some other classes of languages~\cite{DBLP:conf/fct/CzerwinskiMRZ15}.
This inspired us to start a quest for decidable cases beyond regular languages.

In~\cite{DBLP:journals/corr/ClementeCLP16} we have shown decidability of
\emph{unary separability} of reachability sets of vector addition systems (\vases).
By \emph{unary sets} we mean Parikh images of commutative regular languages,
and thus the latter problem is equivalent to commutative regular separability of (commutative closures of) \vas languages.
The decidability status of the regular{\sepsep}separability problem for the whole class of \vas languages remains open.

%
%However adding a stack is not the only way to extend finite automata,
%one possible another direction would be adding one or more weak counters (without tests for zero),
%thus obtaining One Counter Nets or the more general Vector Addition System with States (VASSs).
%We conjecture that regular{\sepsep}separability of languages of VASSs is decidable. On the way to this result we consider
%separability of VASS reachability set by \emph{modular sets} and the more general \emph{unary sets}.
%As we shown in Theorem~\ref{thm:reg-sep-pi-vas} it can be used for example
%to show decidability of regular{\sepsep}separability of commutative closures of languages of Vector Addition Systems (VASs, without states).
%We also believe that unary and modular{\sepsep}separability of VASS reachability set is a natural problem and is interesting in its own.
%Unary sets are exactly the ones, which can be recognized by a finite automaton, which inputs tuples from $\N^d$,
%so they somehow correspond to regular languages.
%\lorenzo{unary sets are the recognizable subsets of $\N^d$ and semilinear sets are the rational subsets of $\N^d$}
%Concretely finite automaton can count on every coordinate exactly till some threshold value and beyond this value modulo
%some number. A unary set is a finite union of products of one dimensional sets, which are periodic beyond some threshold.
%A modular set is conceptually simpler, but similar, it is a finite union of products of periodic one dimensional sets.
%

\myparagraph{Our contribution}
This paper is a continuation of the line of research trying to understand the regular separability problem for language classes 
beyond regular languages.
We report a further progress towards solving the open problem mentioned above: 
we show decidability of the regular{\sepsep}separability problem for the subclass of {\vas} languages
where we allow negative counter values during a run.
This class of languages is also known as languages of \emph{integer \vasses},
and it admits many different characterizations;
for instance, it coincides with languages of \emph{one-way reversal-bounded counter machines}~\cite{Ibarra78},
\emph{Parikh automata}~\cite{DBLP:conf/icalp/KlaedtkeR03}
(cf. also %\cite[Section A.3]{KlaedtkeR:Techrep:02} and
\cite[Proposition 11]{DBLP:conf/ncma/CadilhacFM11}),
which in turn are equivalent to the very similar model of \emph{constrained automata}~\cite{DBLP:journals/ijfcs/CadilhacFM13}.
In this paper, we present our results in terms of constrained automata,
but given the similarity with Parikh automata (and in light of their equivalence),
we overload the name Parikh automata for both models.

Notice that PA languages are not closed under complement,
and thus our decidability result about regular separability does not imply decidability of the regularity problem
(is the language of a given Parikh automaton regular?).
Moreover, the regularity problem for PA languages is actually \emph{undecidable} \cite{DBLP:conf/ncma/CadilhacFM11}%
\footnote{Later shown decidable for unambiguous PA~\cite{DBLP:journals/ijfcs/CadilhacFM13}.},
which makes our decidability result a rare instance of a case where regularity is undecidable but regular separability is decidable.
A result in a similar spirit is that piecewise testability of a context-free language is undecidable,
while piecewise-testable separability of two context-free languages is decidable \cite{DBLP:conf/fct/CzerwinskiMRZ15}.

Parikh automata %(called also constrained automata)
are finite nondeterministic automata
where accepting runs are further restricted to satisfy a semilinear condition on the multiset of transitions appearing in the run.
Our decidability result is actually stated in the more general setting of \emph{$\C$-Parikh automata},
where $\C \subseteq \bigcup_{d \in \N} \powset{\N^d}$ is a class of sets of
vectors used as an acceptance condition.
We prove that the regular{\sepsep}separability problem for languages of $\C$-Parikh automata
reduces to the \emph{unary} separability problem for the class $\C$ itself,
provided that $\C$ is effectively closed under inverse images of affine functions.
Two prototypical classes $\C$ satisfying the latter closure condition are semilinear sets and \vas reachability sets.
Moreover, unary separability of semilinear set is known to be decidable~\cite{DBLP:journals/ipl/ChoffrutG06},
and as recalled before the same result has recently been extended to \vas reachability sets~\cite{DBLP:journals/corr/ClementeCLP16}.
As a consequence of our reduction,
we thus deduce decidability of regular separability of $\C$-Parikh automata languages
where the acceptance condition $\C$ can be instantiated to either the semilinear sets, or the \vas reachability sets. %of the two cases above.
%
% Our algorithm is elementary under the condition that operations on sets from $\C$ can be performed in elementary time.

%
%\myparagraph{Related research}
%Parikh automata have been introduced in~\cite{DBLP:conf/icalp/KlaedtkeR03}
%and later investigated e.g.~in~\cite{DBLP:conf/ncma/CadilhacFM11,DBLP:journals/ijfcs/CadilhacFM13}.
%Languages of Parikh automata have decidable emptiness, but they are not closed under complement
%and and have undecidable universality and regularity problems. %\wojtek{check it}
%In~\cite{DBLP:journals/ijfcs/CadilhacFM13} class of languages of unambiguous Parikh automata was consider.
%This work has shown that the class enjoys nice closure properties (among others closure under boolean combinations) and
%decidability of several classical problems (emptiness, universality, finiteness, inclusion and regularity).
%

% !TEX root = main.tex

\section{Preliminaries}
%\myparagraph{Notations} In the following for a set $E$ we will $\overline{E}$ the complement of $E$. For a vector $u\in \N^k$, and 
%$j\in\{1,\ldots,k\}$ we denote by $u[j]$ the $j$th coordinate of $u$. 
%\charles{Maybe useless paragraph, delete if you don't like (but it does not hurt I think)}

%\myparagraph{Finite automata}
%We often write a transition $(p, a, q) \in Q \times \Sigma \times Q$ as $p \trans{a} q$.
%A run of $\A$ over a word $w = a_1 \cdots a_n$ is a sequence of transitions
%$p_0 \trans{a_1} p_1 \trans{a_2} \cdots \trans{a_n} p_n$, it is \emph{accepting} if $p_0$ is initial and $p_n$ is final.
%The language of $\A$, denoted $L(\A)$, is the set of words $w$ such that there exists an accepting run over $w$.

%We say that $\A$ is \emph{unambiguous} if for every $w \in L(\A)$ there exists exactly one accepting run over $w$.
%Unambiguous automata are widely studied, see for example a recent survey~\cite{DBLP:conf/dcfs/Colcombet15}.

\myparagraph{Vectors sets}
A set $S \subseteq \N^d$ is \emph{linear} if there exist a \emph{base} $b \in \N^d$
and \emph{periods} $p_1, \ldots, p_k \in \N^d$ s.t. $S = \{b + n_1 p_1 + \ldots + n_k p_k \mid n_1, \ldots, n_k \in \N\}$,
and it is \emph{semilinear} if it is a finite union of linear sets.
%
%Semilinear set are also known as \emph{rational set} of vectors, that is the closure under union, product and Kleene star. 
For a vector $v \in \N^d$ and $i \in \set {1, \dots, d}$, let $v[i]$ denote its $i$-th coordinate.
For $n\in\N$, we say that two vectors $x, y \in \N^d$ are \emph{$n$-unary equivalent},
written $x \equiv_n y$,
if for every coordinate $i \in \set {1, \dots, d}$ it holds $x[i] \equiv y[i] \mod n$
and moreover $x[i] \leq n \iff y[i] \leq n$.
A set $S \subseteq \N^d$ is \emph{unary} if for some $n$, $S$ is a union of equivalence classes of $\equiv_n$.
Intuitively, to decide membership in a unary set $S$
it is enough to count on every coordinate exactly up to some threshold $n$,
and modulo $n$ for values larger than $n$.
Every unary set is in particular semilinear. 
%Unary set are also known as set \emph{recognizable} by finite monoids.
%The relations between semilinear (rational) and unary (recognizable) set has been studied~(see for instance~\cite{DBLP:journals/ipl/ChoffrutG06}).

Let $\Sigma = \set {a_1, \ldots, a_k}$ be an ordered alphabet.
For a word $w \in \Sigma^*$ and a letter $a_i \in \Sigma$, by $\#_{a_i}(w)$ we denote the number of letters $a_i$ in $w$.
The \emph{Parikh image} of a word $w \in \Sigma^*$ is the vector $\parim(w) = (\#_{a_1}(w), \ldots, \#_{a_k}(w)) \in \N^k$. 
The \emph{Parikh image} of a language $L \subseteq \Sigma^*$
is $\parim(L) = \{\parim(w) \mid w \in L\}$, the set of Parikh images of all words belonging to $L$.
% The well known Parikh theorem states that the Parikh image of a context-free language is always semilinear.

\myparagraph{Parikh automata}
A \emph{nondeterministic finite automaton} (NFA) $\A = (Q, I, F, T)$ over a finite alphabet $\Sigma$
consists of a finite set of states $Q$, distinguished subsets
of initial and final states $I, F \subseteq Q$, and a set of transitions $T \subseteq Q \times \Sigma \times Q$.
A \emph{nodeterministic Parikh automaton}%
\footnote{This is the same as \emph{constrained automata} from \cite{DBLP:journals/ijfcs/CadilhacFM13}}
is a pair $(\A, S)$ consisting of an NFA $\A$ and a semilinear set $S \subseteq \N^d$,
for $d = |T|$. 
A \emph{run} of a Parikh automaton over a word $w = a_1 \ldots a_n \in\Sigma^*$ is a sequence of transitions
$\rho = t_1 \ldots t_n \in T^*$, where $t_i = (q_{i-1}, a_i, q_{i})$, starting in an initial state $q_0$.
A run $\rho$ is \emph{accepting} if its ending state $q_n$ is final and $\parim(\rho) \in S$.
The language of a Parikh automaton, denoted $L(\A, S)$, contains all words $w$ admitting an accepting run;
it is thus a subset of the language $L(\A)$ of the underlying NFA.

\begin{remark}
	A more liberal definition for $\emptyword$-Parikh automata can be given
	by allowing transitions to read $\emptyword$'s,
	i.e, $T \subseteq Q \times (\Sigma \cup \set \emptyword) \times Q$.
	However, allowing $\emptyword$-transitions does not increase the expressiveness of Parikh automata,
	which follows from closure under (possibly erasing) homomorphisms of the latter class \cite[Property 4.(2)]{DBLP:conf/icalp/KlaedtkeR03}.
\end{remark}

One can generalize Parikh automata by using some other family of vector sets in the place of semilinear sets. 
For a class $\C \subseteq \bigcup_{d \in \N} \powset{\N^d}$ of vector sets,
a \emph{$\C$-Parikh automaton} is a pair $(\A, S)$, where $\A$ is an NFA and $S \in \C$. 
The language $L(\A, S)$ is then defined exactly as above. 

A $\C$-Parikh automaton $(\A, S)$ is \emph{deterministic} if the underlaying automaton $\A$ is so.
The languages of (non)deterministic $\C$-Parikh automata are shortly called (non)deterministic $\C$-Parikh languages below.
%Note that thus the unambiguity does not depend in any way on the set $S$. 

\section{Main result}

A function $f: \N^k \to \N^\ell$ is called \emph{affine} if it is of the form $f(v) = Mv + u$ for a  matrix $M$ of dimension
${\ell\times k}$ and a vector $u \in \N^\ell$.
A class of vector sets $\C \subseteq \bigcup_{d \in \N} \powset{\N^d}$ is called \emph{robust} if it fulfills the following two conditions:
\begin{itemize}
  \item $\C$ is effectively closed under inverse images of affine functions, 
  \item the unary{\sepsep}separability problem is decidable for $\C$.
\end{itemize}

As our main result we prove decidability of the regular{\sepsep}separability problem for $\C$-Parikh automata.
%The decidability is surprising, contrasted with unecidability of the regularity problem, even for plain Parikh automata. 
%\myparagraph{Main result}
%In~\cite{DBLP:conf/ncma/CadilhacFM11} it is shown that regularity problem is decidable for unambiguous Parikh automata
%languages. Unambiguous Parikh automata languages are effectively closed under complement~\cite{DBLP:journals/ijfcs/CadilhacFM13},
%so separability problem is a generalization of regularity problem in that case.
%Thus our main result, Theorem~\ref{thm:reg-sep}, is a generalization of that in two directions: from semilinear sets to any robust
%class of sets and, more importantly, from regularity problem to separability problem.

\begin{theorem}\label{thm:reg-sep}
The regular{\sepsep}separability problem is decidable for $\C$-Parikh automata, for every robust class $\C$ of vector sets.
\end{theorem}
The proof of Theorem~\ref{thm:reg-sep} is split into two parts.
In Section~\ref{sec:nondet2det} we provide a reduction of 
the regular separability problem of \emph{nondeterministic} $\C$-Parikh automata to the same problem of \emph{deterministic} ones;
this step is crucial for understanding how the regular{\sepsep}separability problem differs from the regularity problem, 
which does not admit a similar reduction.
Then in Section~\ref{sec:proof} we reduce the regular{\sepsep}separability problem for deterministic $\C$-Parikh automata to
the unary{\sepsep}separability problem for vector sets in $\C$.

In Section~\ref{sec:dec} we consider two instantiations of the class $\C$.
First, taking semilinear sets as $\C$ we derive decidability for plain Parikh automata.
Second, we consider the class $\Cvas$ of \emph{sections of} reachability sets of \vases
(detailed definitions are deferred to Section~\ref{sec:dec}), which allows us to obtain decidability for $\Cvas$-Parikh automata.
Note that the latter model properly extends plain Parikh automata.

\section{From nondeterministic to deterministic PA}  \label{sec:nondet2det}

\newcommand{\im}[1]{\text{\sc Im}(#1)}

\newcommand{\image}[2]{{#1}(#2)}
\newcommand{\invimage}[2]{{#1^{-1}}(#2)}

The aim of this section is to prove the following lemma:
\begin{lemma} \label{l:nondet2det}
	If $\C$ is closed under inverse images of linear mappings,
	then the regular{\sepsep}separability problem of nondeterministic $\C$-Parikh automata
	effectively reduces to the same problem of deterministic ones. 
\end{lemma}
Before embarking on the proof, we need to state and prove a couple of auxiliary facts.
In the rest of the section, we assume that the class $\C$ is closed under inverse images of linear mappings.
Given two alphabets $\Sigma$ and $\Gamma$,
a \emph{letter-to-letter homomorphism} is a function $h : \Sigma \to \Gamma$
which extends homomorphically to a function from $\Sigma^*$ to $\Gamma^*$,
and thus to languages.

%In this section we consider surjective letter-to-letter word homomorphisms, i.e., homomorphisms induced by a surjective function between alphabets,
%and call these homomorphisms \emph{projections}. 
%For a class of languages $\G$, 
%let $\im{\G}$ denote the class of all images of projections of languages from $\G$, that is $\im{G} = \setof{\image{h}{L}}{h \text{ a projection, } L\in\G}$.

%The result of this section, as stated in Lemma~\ref{l:nondet2det} below, is general and therefore presented quite abstractly. 
%However, our intention is to apply the result to nondeterministic $\C$-Parikh automata, as stated in Corollary~\ref{cor:Parikh} below.

\begin{lemma}  \label{l:det}
	Every nondeterministic $\C$-Parikh language is the image of a letter-to-letter homomorphism
	of a deterministic $\C$-Parikh language.
%	$L$ can be effectively presented as a projection $\image{h}{L'}$ 
%	of a deterministic $\C$-Parikh language $L'$.
\end{lemma}

\begin{proof}
	Fix a nondeterministic $\C$-Parikh automaton $(\A, S)$
	of maximal nondeterministic branching $n$
	recognizing the language $L(\A, S) \subseteq \Sigma^*$. % with $|\Sigma| = d$.
%	We claim that one can compute a projection $h$ and a deterministic $\C$-Parikh automaton $(\A', S)$ with the same constraint $S\in\C$, that recognises a language $L'$
%	such that  $L = \image{h}{L'}$.
	Consider the extended alphabet $\Gamma = \Sigma \times \set {1, \ldots, n}$
	obtained by labelling each symbol from $\Sigma$ with an index to resolve nondeterminism,
	and consider the letter-to-letter homomorphism $h : \Gamma \to \Sigma$ that maps $(a, i)$ to $a$.
	Let $(\B, T)$ be the deterministic $\C$-Parikh automaton over $\Gamma$
	which is obtained from $\A$ by a relabelling in every state
	the $i$-th transition over $a$ by $(a, i)$.
	The acceptance condition $T \subseteq \N^{|\Sigma| \cdot n}$
	is $T := \phi^{-1}(S)$,
	where $\phi : \N^{|\Sigma| \cdot n} \to \N^{|\Sigma|}$
	is the linear mapping that sums up the entries $(a, 1), \dots, (a, n)$ corresponding to the original symbol $a$.
	One easily verifies that $L(\A, S) = \image{h}{L(\B, T)}$, as required.
	\qed
\end{proof}
%
%In the proof of the following lemma we will use robustness of $\C$.
%
\begin{lemma} \label{l:closuredet}
	Deterministic $\C$-Parikh languages are effectively closed under inverse images of letter-to-letter homomorphisms. % of projections.
\end{lemma}
\begin{proof}
	Given a deterministic $\C$-Parikh automaton $(\A, S)$ over $\Sigma$
	and a letter-to-letter homomorphism $h : \Gamma \to \Sigma$, one computes a deterministic
	$\C$-Parikh automaton $(\B, T)$ as follows.
	The automaton $\B$ is obtained by replacing every transition $(p, a, q)$ in $\A$ %, for $a\in \Sigma$, 
	by transitions $(p, b, q)$, one for every $b\in \invimage{h}{a}$.
	The constraint $T\in\C$ is 
	the inverse image of $S$ under the linear function % from $f : \N^{k\times n} \to \N^k$ 
	that sums up values on all coordinates corresponding to letters in $\invimage{h}{a} \subseteq \Gamma$
	in order to compute the value on the coordinate corresponding to $a \in \Sigma$.
	%\[
	%(x_1, x_2, \ldots, x_{k n}) \ \stackrel{f}{\longmapsto} \ 
	%(x_1 + \ldots + x_n, \ x_{n+1} + \ldots + x_{2n}, \ \ldots, \ x_{(k-1)n + 1} + \ldots + x_{k n}).
	%\]
	Finally, the constraint $T$, and hence also the automaton $(\B, T)$ can be computed.
	\qed
\end{proof}
%
%Languages are typically represented by their recognizing devices.
%We say that a class $\G$ of languages is \emph{effectively closed under inverse images of projections} if for a representation of a language $L\in \G$ and a projection,
%one can compute in polynomial time a representation of the inverse image $\invimage{h}{L}$. 
%We observe:
%
\begin{lemma} \label{l:closure}
%Projection images of 
	Nondeterministic $\C$-Parikh languages are effectively closed under inverse images of letter-to-letter homomorphisms. % of projections.
\end{lemma}
\begin{proof}
	The construction is exactly the same as in the proof of Lemma~\ref{l:closuredet} above,
	but the resulting automaton does not have to be deterministic.
%Let $L$ be a deterministic $\C$-Parikh language, and let $g, h$ be two projections with the same co-domain.
%The inverse image $\invimage{g}{\image{h}{L}}$ of the image of a language, after a suitable modification of the two projections, 
%is the image $\image{h'}{\invimage{g'}{L}}$ of the inverse image of the same language. 
%The modified projections $g', h'$ can be computed in polynomial time. Moreover, a deterministic $\C$-Parikh automaton recognizing
%$\invimage{g'}{L}$ can be computed from a deterministic $\C$-Parikh automaton recognising $L$, by Lemma~\ref{l:closuredet}.
\qed
\end{proof}
The next lemma is the cornerstone of our reduction.
It allows to make one automaton deterministic without introducing nondeterminism in the second one.
\begin{lemma} \label{l:red}
	Languages $\image{h}{L}$ and $K$ are regular{\sepsep}separable if, and only if, $L$ and $\invimage{h}{K}$ are so.
\end{lemma}
\begin{proof}
%Both implications are an easy check.

	For the ``only if'' direction,
	if a regular language $R$ separates $\image{h}{L}$ and $K$ then the language $\invimage{h}{R}$ 
	separates $L$ and $\invimage{h}{K}$.
	Indeed, the inclusion $L \subseteq \invimage{h}{R}$ follows from the inclusion $\image{h}{L} \subseteq R$
	since $L \subseteq \invimage{h}{\image{h}{L}}$,
	and the disjointness of
	$\invimage{h}{R}$ and $\invimage{h}{K}$ follows from disjointness of $R$ and $K$.

	For the ``if'' direction,
	if a regular language $R$ separates $L$ and $\invimage{h}{K}$ then the language 
	$\image{h}{R}$ separates the languages $\image{h}{L}$ and $K$.
	The inclusion $\image{h}{L} \subseteq \image{h}{R}$ follows by the inclusion $L \subseteq R$, and 
	the disjointness of $\image{h}{R}$ and $K$ follows from disjointness of $R$ and $\invimage{h}{K}$
	since $\image{h}{\invimage{h}{K}} \subseteq K$.
	\qed
\end{proof}
%
%\begin{proof}
\myparagraph{Proof of Lemma~\ref{l:nondet2det}}
	Let $L, K$ be two nondeterministic $\C$-PA languages.
	By Lemma~\ref{l:det}, we may assume that $L$ is the image $\image{h}{L_1}$ of a deterministic language $L_1$.
	By Lemma~\ref{l:red}, regular separability for $\image{h}{L_1}, K$ is the same as for $L_1, \invimage{h}{K}$.
	By Lemma~\ref{l:closure}, $\invimage{h}{K}$ is a nondeterministic language itself,
	so by Lemma~\ref{l:det} it equals the image $\image{g}{K_1}$ of a deterministic language $K_1$.
	We have thus reduced to regular separability for $L_1, \image{g}{K_1}$,
	where now both $L_1$ and $K_1$ are deterministic languages.
	Since regular{\sepsep}separability is symmetric,
	regular separability for $L_1, \image{g}{K_1}$ is the same for $\image{g}{K_1}, L_1$.
	Applying once more Lemma~\ref{l:red},
	the latter statement is equivalent to regular separability for $K_1, \invimage{g}{L_1}$.
	By Lemma~\ref{l:closuredet}, $\invimage{g}{L_1}$ is a deterministic language.
	Since every step was effective, this concludes the proof.
	\qed
%\end{proof}
%

% !TEX root = main.tex

\section{Regular{\sepsep}separability reduces to unary{\sepsep}separability}  \label{sec:proof}

In this section we reduce regular{\sepsep}separability of deterministic $\C$-Parikh languages to unary{\sepsep}separability of
vector sets in $\C$.
\begin{lemma}  \label{l:deciddet}
	Let $\C$ be a class of vectors closed under inverse images of affine mappings.
	The regular{\sepsep}separability problem for deterministic $\C$-Parikh automata
	reduces to the unary{\sepsep}separability problem for vector sets in $\C$.
\end{lemma}

The rest of this section is devoted to the proof of this lemma.
Let $L_1, L_2 \subseteq \Sigma^*$ be languages of deterministic $\C$-Parikh automata $(\A_1, S_1)$ and $(\A_2, S_2)$, respectively.
The proof comprises three steps:
\begin{enumerate}
	\item As the first step, we show that w.l.o.g.~we may assume $\A_1 = \A_2$.
	\item In the second step, we partition $\Sigma^*$ into finitely many regular languages $K_1, \ldots, K_m$
	and we reduce regular separability of $L_1$ and $L_2$
	to regular separability of $L_1 \mycap K_i$ and $L_2 \mycap K_i$ for every $i \in \set{1, \ldots, m}$.
	These subproblems turn out to be easier that the general one,
	due to the additional structural information encoded in the languages $K_i$'s.
%In a very easy, third step we show that for $i \neq j$ the languages $L_1 \mycap K_i$ and $L_2 \mycap K_j$ are always
%separable by some regular language. 
	\item In the last step, we reduce separability of $L_1 \mycap K_i$ and $L_2 \mycap K_i$
	to unary{\sepsep}separability of vector sets in $\C$.
% for $L_1 \mycap K_i$ and $L_2 \mycap K_i$, 
\end{enumerate}

\myparagraph{Step 1: Unifying the underlying automaton}
As the input languages are subsets of regular languages recognised by their underlying finite automata,
$L_1 = L(\A_1, S_1) \subseteq L(\A_1)$ and $L_2 = L(\A_2, S_2) \subseteq L(\A_2)$,
it is enough to consider separability of $L_1$ and $L_2$ \emph{inside} the intersection of $L(\A_1)$ and $L(\A_2)$:

\begin{proposition}\label{prop:intersection}
	The languages $L_1$ and $L_2$ are regular{\sepsep}separable if, and only if,
	the languages $L_1 \mycap L(\A_2)$ and $L_2 \mycap L(\A_1)$
	are so.
\end{proposition}

\begin{proof}
The ``only if'' direction is trivial as every language separating $L_1$ and $L_2$ separates
$L_1 \mycap L(\A_2)$ and $L_2 \mycap L(\A_1)$ as well.
For the opposite direction, we observe that if a regular language $S$ separates $L_1 \mycap L(\A_2)$ and $L_2 \mycap L(\A_1)$,
then $S' = S \cup \overline{L(\A_2)}$ is a regular language separating $L_1$ and $L_2$.
\qed
\end{proof}

Let $\A$ be the product automaton of $\A_1$ and $\A_2$, and thus $L (\A) = L(\A_1) \mycap L(\A_2)$. 
It is deterministic since both $\A_1$ and $\A_2$ are so.
%
%It is defined by a standard product construction.
%Let $\A_1 = (Q_1, q_1, F_1, T_1)$ and $\A_2 = (Q_2, q_2, F_2, T_2)$.
%Then $\A = (Q, q, F, T)$, where $Q = Q_1 \times Q_2$, $q = (q_1, q_2)$, $F = F_1 \times F_2$
%and $T = \{(p_1, p_2) \trans{a} (p'_1, p'_2) \mid p_1 \trans{a} p'_1 \wedge p_2 \trans{a} p'_2\}$.
%For every transition $t = (p_1, p_2) \trans{a} (p'_1, p'_2) \in T$ let $\phi_1(t) = p_1 \trans{a} p'_1 \in T_1$
%and analogously we define $\phi_2(t) \in T_2$.
%It is easy to verify that indeed $L(\A) = L(\A_1) \mycap L(\A_2)$. Moreover $\A$ is unambiguous.
%Indeed, if there is an accepting run over $w$ in $\A$, then there is exactly one
%accepting run over $w$ in $\A_1$ and similarly in $\A_2$. This means that there is exactly one
%accepting run over $w$ in $\A$.
%
%Therefore our task it to decide regular separability of $L_1 \mycap L(\A)$ and $L_2 \mycap L(\A)$.
%The last part of the first step is to show 
We claim that one can compute sets $U_1, U_2 \in \C$ such
that $L_1 \mycap L(\A_2) = L(\A, U_1)$ and $L_2 \mycap L(\A_1) = L(\A, U_2)$.
The set $T$ of transitions of $\A$ is a subset of the product $T_1 \times T_2$ of transitions of $\A_1$ and $\A_2$, and thus there are obvious projections functions
$\pi_1 : T\to T_1$ and $\pi_2 : T \to T_2$.
If we enumerate the transition sets,
say $T_1 = \{t_1^1, \ldots, t_1^m\}$,
$T_2 = \{t_2^1, \ldots, t_2^n\}$,
and $T = \{t_1, \ldots, t_\ell\}$ with $\ell \leq m \cdot n$,
we obtain
$\pi_1 : \set{1, \ldots, \ell} \to \set{1, \ldots, m}$ and $\pi_2 : \set{1, \ldots, \ell} \to \set{1, \ldots, n}$.
We use these projections to define two linear (and in particular, affine) functions
$\psi_1: \N^\ell \to \N^m$ and $\psi_2: \N^\ell \to \N^n$ 
which instead of counting transitions in $T$, count the corresponding transitions in $T_1$ or in $T_2$, respectively;
formally,
\[
  \psi_1(v)[j] = \sum_{i: \pi_1(i) = j} v[i] \qquad
  \psi_2(v)[j] = \sum_{i: \pi_2(i) = j} v[i].
\]
% In other words vector $\psi_1(v)$
%on the coordinate $j$ counts number of occurrences of transitions $t_i \in T$ in $v$
%which correspond to the transition $t_j \in T_1$.
%Similarly we define an affine function $\psi_2: \N^\ell \to \N^n$.
Finally, we set $U_1 := \psi_1^{-1}(S_1)$ and $U_2 := \psi_2^{-1}(S_2)$.
Intuitively, $U_1$ and $U_2$ are as $S_1$ and $S_2$, except that instead of single transitions of $\A_1$ or $\A_2$ they are seeing pairs of transitions,
and simply ignore one of them.
Since $\C$ is closed under inverse images of affine mappings by assumption,
$U_1, U_2 \in \C$.
For the rest of the proof we may thus assume that the input automata are $(\A, U_1)$ and $(\A, U_2)$.

%\begin{remark}
%	\label{rem:complexity:1}
%	In this remark and in Remarks~\ref{rem:complexity:2}, \ref{rem:complexity:3}, \ref{rem:complexity:semilinear},
%	\ref{rem:complexity:VAS} we analyze the complexity of our algorithm.
%	This step of the reduction requires taking two inverse images of affine linear mappings
%	and the cartesian product of two automata.
%	%
%\end{remark}

\myparagraph{Step 2: Regular partition using skeletons}
We now define a partition of $\Sigma^*$ into finitely many parts, such that words belonging to the same part
behave similarly with respect to automaton $\A$.

We use the notion of \emph{skeleton} of a run, defined already in~\cite{DBLP:journals/ijfcs/CadilhacFM13},
where it was used to solve the regularity problem of \emph{unambiguous} Parikh automata.
Consider a run $\rho = t_1 \cdots t_k \in T^*$. The idea of skeleton is to traverse $\rho$ from left to right and remove loops,
but only if such removal does not decrease the set of states visited so far.
% that we read a run from left
%to right, we keep adding transitions to the skeleton, but if we reach a state visited before, we have a loop and
%remove it from the skeleton.
%However, we remove loops only in the case when after the removal set of states visited in the skeleton
%remains unchanged. 
%
Formally, the skeleton is a function from runs to runs defined by induction. We set $\skel(\varepsilon) = \varepsilon$.
For the induction step, suppose that $\skel(t_1 \ldots t_{k-1}) = u_1 \ldots u_\ell \in T^*$ is already defined, and
let $q$ be the ending state of the new transition $t_k$. If $q$ does not appear in the run $u_1 \ldots u_\ell$, we put
$\skel(t_1 \ldots t_k) = u_1 \ldots u_\ell t_k$. Otherwise, let $u_m$, for $m < \ell$, be the last transition that ends in state $q$.
If all states visited by $u_{m+1} \ldots u_\ell$ are also visited by $u_1 \ldots u_m$, we put 
$\skel(t_1 \ldots t_k) = u_1 \ldots u_m$ thus removing the loop; otherwise, we put $\skel(t_1 \ldots t_k) = u_1 \ldots u_\ell t_k$.

%and  $\skel(t_1 \cdots t_k) = \skel(t_1 \cdots t_{k-1}) t_k$ if there is no loop in $\skel(t_1 \cdots t_{k-1}) t_k$
%which can be removed without affecting the set of visited states. In the other case we set $\skel(t_1 \cdots t_k) = \sigma$,
%where $\sigma$ is obtained from $\skel(t_1 \cdots t_{k-1}) t_k$ by removing the loop between the state visited after $t_k$
%and its last previous occurrence. 

The so defined skeleton $\skel(\rho)$ of a run $\rho$ has two properties:
1) $\skel(\rho)$ visits the same states as $\rho$, 2) the length of $\skel(\rho)$ is at most $n^2$, where $n$ is the number
of states in the automaton $\A$. The first point is clear by definition.
In order to see the second point,
assume towards a contradiction that the length of the skeleton is longer than $n^2$.
By the pigeonhole principle, some state is thus visited more than $n$ times, so there are at least $n$
loops in between two consecutive occurrences of this state in the skeleton. Therefore it is impossible that each loop contains
some new state not present in all the previous loops, and thus
%Indeed, because it would mean that the last loop contains at least $n+1$ states. 
one of these loops should be removed during the process of creating the skeleton, a contradiction.

We abusively call a run $\rho$ a \emph{skeleton} if $\skel(\rho) = \rho$.
Because of the bound $n^2$ on the length of a skeleton,
if $d$ is the total number of transitions of $\A$,
then there are at most $d^{n^2}$ skeleton runs.
Let $\rho_1, \ldots, \rho_m$ be all the skeletons, with $m \leq k^{n^2}$.
We define $K_i$ to be the set of all words $w$ having an accepting run $\rho$
in automaton $\A$ with $\skel(\rho) = \rho_i$.
Since $\A$ is deterministic we know that $K_i \mycap K_j = \emptyset$ for $i \neq j$. 
% Indeed, otherwise a word $w \in K_i \mycap K_j$ would have at least two different runs in $\A$: one with skeleton $\rho_i$ and  one with skeleton $\rho_j$. 
Therefore $K_1, \ldots, K_m$ and $K_{m+1} = \Sigma^* \setminus (\bigcup_{1 \leq i \leq m} K_i)$
form a partition of $\Sigma^*$.
%Notice here that our argument (and in fact the whole proof) actually works in a bit larger
%class than unambiguous Parikh automata. We may allow that a word have many runs, but we need a condition that
%all of them have the same skeleton.
% (I'm not sure anymore about the above remark)
%
All languages $K_i$ are necessarily regular, since the skeleton can be computed by a finite automaton.
% Indeed, it is easy to construct an automaton $\B$, which recognizes set
%of words, for which there exists a run in $\A$ with skeleton $\rho_i$. Automaton $\B$ for a prefix of a word $w$
%remembers set of the skeletons of the runs of $\A$ over the prefix it has red. When $\B$ will read whole word $w$
%only one of these skeletons will end in a final state. If this is $\rho_i$ automaton $\B$ should accept, otherwise it should reject.
%Because of the fact that there are only finitely many skeletons of $\A$ automaton $\B$ can keep this information in a finite
%set of states.

We state the following lemma, which can be seen
as generalization of Proposition~\ref{prop:intersection}.

\begin{lemma}
Let $K_1, \ldots, K_k$ be regular languages forming a partition of $\Sigma^*$.
Two languages $L_1, L_2 \subseteq \Sigma^*$ are regular{\sepsep}separable if, and only if,
$L_1 \mycap K_i$ and $L_2 \mycap K_i$ are regular separable for all $i \in \{1, \ldots, k\}$.
\end{lemma}

\begin{proof}
The ``only if'' direction is trivial, since every language separating $L_1$ and $L_2$ separates
$L_1 \mycap K_i$ and $L_2 \mycap K_i$ as well.
For the opposite direction, we observe that if for every $i$ the languages $L_1 \mycap K_i$ and $L_2 \mycap K_i$
are separable by a regular language $S_i$, then $L_1$ and $L_2$ are separable by the regular language
$S = \bigcup_{1 \leq i \leq k} (S_i \mycap K_i)$.
\qed
\end{proof}

%\begin{remark}
%	\label{rem:complexity:2}
%	From the previous bound on the number of skeleton runs,
%	we have reduce a separability problem
%	to at most $d^{n^2}$ many separability problems (of a special form),
%	where $d$ is the total number of transitions of the automaton,
%	and $n$ is the total number of states of the automaton.
%\end{remark}

Therefore, it only remains to decide regular separability for the languages
$L(\A, U_1) \mycap K_i$ and $L(\A, U_2) \mycap K_i$.

\myparagraph{Step 3: Reduction to unary{\sepsep}separability in $\C$}
Fix a skeleton $\rho_i$.
%  and let $Q_v \subseteq Q$ be all the states which are visited along $\rho_i$. 
Let $c_1, \ldots, c_m$ be all the simple cycles in the automaton $\A$ which visit
only states visited by $\rho_i$. 
Since a cycle cannot visit the same state twice (except the initial state),
it has length at most $n$,
and thus the number of simple cycles is $m \leq d^n$,
where $d$ is the number of transitions of the automaton.
Notice that any run $\rho$ with $\skel(\rho) = \rho_i$ decomposes into the skeleton $\rho_i$
and a bunch of simple cycles from $\{c_1, \ldots, c_m\}$.
Let $T = \set{t_1, \ldots, t_d}$, thus $\rho \in T^*$.
Let $\mu: \N^m \to \N^d$ be the affine function that transforms counting cycles into counting transitions,
which is defined as
\begin{align*}
	\mu(x_1, \ldots, x_m) = \parim(\rho_i) + \sum_{1 \leq i \leq m} \parim(c_i) \cdot x_i.
\end{align*}
%
%for coordinate-wise addition.
%
(Notice that the function above is affine, and not linear,
since it requires to take into account the initial cost of the skeleton $\parim(\rho_i)$.)
In other words,
$\mu(x_1, \ldots, x_m)$ returns Parikh image of a run which decomposes into the skeleton $\rho_i$ and
$x_i$ cycles $c_i$, for every $i$.
Let $V_1 = \mu^{-1}(U_1)$ and $V_2 = \mu^{-1}(U_2)$ be the corresponding sets counting cycles instead of transitions.
Since $\C$ is closed under the inverse image of affine mappings,
$V_1, V_2 \in \C$.
%Intuition behind $V_1 \subseteq \N^m$ is that it contains
%all the possible number of occurrences of cycles, which imply that the run with a skeleton $\rho_i$ has a Parikh image in $U_1$.
%Similarly for $V_2$. The argument finishing our reasoning is included in the following lemma.

\begin{lemma}  \label{l:tounary}
The following two conditions are equivalent:
	\begin{enumerate}
	  \item The two languages $L(\A, U_1) \mycap K_i, L(\A, U_2) \mycap K_i \subseteq \Sigma^*$ are regular{\sepsep}separable.
	  \item The two sets of vectors $V_1, V_2 \subseteq \N^m$ are unary{\sepsep}separable.
	\end{enumerate}
\end{lemma}

\begin{proof}
	For the implication $1) \Rightarrow 2)$, suppose $R$ is a regular language separating $L(\A, U_1) \mycap K_i$ and $L(\A, U_2) \mycap K_i$.
	Fix $\omega \in \N$ such that for all words $x, y, z \in \Sigma^*$,
	\begin{equation}\label{eq:omega-power}
	x y^\omega z \in R \iff x y^{2\omega} z \in R.
	\end{equation}
	It is easy to see that for every regular language $R$ such $\omega$ exists. The simplest way of showing this
	is to consider the syntactic monoid $M$ of $R$ and to let $\omega$ be its idempotent power, i.e., a number
	such that $m^\omega = (m^\omega)^2$ for every $m \in M$.

	Recall $n$-unary equivalence: $u \equiv_n v$
	if for every coordinate $1 \leq i \leq m$ we have $u[i] \equiv v[i] \mod n$ and moreover $u[i] \leq n \iff v[i] \leq n$.
	It is enough to show that for all $v_1 \in V_1, v_2 \in V_2$ it holds $v_1 \not\equiv_\omega v_2$.
	Indeed, if this is the case, the unary set $S = \{v \in \N^m \mid \exists_{v_1 \in V_1} v \equiv_\omega v_1\}$
	separates $V_1$ and $V_2$. 

	Suppose, towards a contradiction, that there are some $v_1 \in V_1, v_2 \in V_2$ such that $v_1 \equiv_\omega v_2$.
	Recall that $c_1, \ldots, c_m$ are all the simple cycles in automaton $\A$ visiting
	only states visited by the skeleton $\rho_i$. For every cycle $c_j$, let's arbitrarily choose a state on it, and let's call it the \emph{fixing state} of $c_j$.
	Let $w_1, \ldots, w_m$ be words labeling the cycles $c_1, \ldots, c_m$, resp., when reading from its fixing state,
	and let $w$ be the word labeling skeleton $\rho_i$.
	Consider a partition $w = s_0 \ldots s_k$ and let $q_d$, for $d \in \{0, \ldots, k-1\}$, be the state,
	which is reached in $\A$ after reading $s_0 \ldots s_d$.
	This partition of $w$ is chosen such that among $q_d$ are all fixing states of cycles $c_1, \ldots, c_m$, every one exactly ones.
	For every $v \in \N^m$ we define a \emph{canonical word} $w_v$ for $v$ as
	the word obtained from pasting into $w$, in places between some $s_d$ and $s_{d+1}$, words $w_1^{v[1]}, \ldots, w_m^{v[m]}$
	in such a way that every $w_j^{v[j]}$ is pasted into the place where its fixing state equals $q_d$ and words
	pasted into the same place are sorted according to indices of the corresponding cycles.
	% the lexicographically smallest interleaving of $m+1$ words:
	% $w$, $w_1^{v[1]}, \ldots, w_m^{v[m]}$, that belongs to $L(\A) \mycap K_i$.

	% $w_v = w\, w_1^{v[1]} \cdots w_m^{v[m]}$. \sltodo{is this word accepted by $\A$?} 
	Notice an important fact: if $v \in V_1$ then 
	 $w_v \in L(\A, U_1) \mycap K_i$, and likewise for $V_2$. Consider words $w_{v_1}$ and $w_{v_2}$.
	One can see that by repeated application of equation~\eqref{eq:omega-power} we can obtain
	that $w_{v_1} \in R \iff w_{v_2} \in R$. But $R$ was supposed to separate 
	$L(\A, U_1) \mycap K_i$ and $L(\A, U_2) \mycap K_i$, a contradiction.

	For proving the implication $2) \Rightarrow 1)$, suppose that a unary set $S$ separates $V_1$ and $V_2$.
	We claim that the language $R = L(\A, \mu(S)) \mycap K_i$ is regular and separates
	$L(\A, U_1) \mycap K_i$ and $L(\A, U_2) \mycap K_i$. 

	We first verify that $R$ separates the languages. Clearly, $U_1 \subseteq \mu(V_1) \subseteq \mu(S)$,
	so $L(\A, U_1) \mycap K_i \subseteq L(\A, \mu(S)) \mycap K_i = R$. The disjointness of $L(\A, U_2) \mycap K_i$ and $R$
	is shown by contradiction. 
	Suppose that there is a word $w\in K_i$ belonging both to $L(\A, \mu(S))$ and to $L(\A, U_2)$, let $\rho$ we run of $\A$
	over $w$ and let $v = \parim(\rho)$.
	We have $v \in \mu(S) \mycap U_2$, which implies $v = \mu(s)$ for some $s \in S \mycap \mu^{-1}(U_2) = S \mycap V_2$.
	In consequence $S \mycap V_2$ is nonempty, thus contradicting  the assumption that $S$ separates $V_1$ and $V_2$.

%	\lorenzo{can this be shown in two steps: 1) $L(\A, T)$ is regular then $T$ is unary (this can be even put in the preliminaries) 2) $\mu(S)$ is itself unary}
	In order to prove that $R$ is regular it suffices to prove that $L(\A, \mu(S))$ is regular.  
	The finite nondeterministic automaton recognizing this language simulates a run $\rho = t_{i_1} \ldots t_{i_\ell}$ of $\A$, and accepts when $\parim(\rho) \in \mu(S)$.
	Since $S$ is unary, the automaton can evaluate this condition using finite memory.
	For every cycle $c_j$, the automaton would store a vector $x_j < \parim(c_j)$, and a number $n_j$
	up to the unary equivalence $\equiv_n$, with the following meaning:
	the vector $\parim(c_i)$ has been already executed $n_j$ times, and $x_i$ is the current ``remainder''.
	Additionally, the automaton stores a vector $x \leq \parim(\rho_i)$ which is counting those transitions on the skeleton
	which have not been counted as cycles.
	At every input letter the automaton guesses nondeterministically one of cycles $c_i$ or the skeleton
	and updates $x_j$, $n_j$ and $x$ accordingly.
	The automaton accepts when $x = \parim(\rho_i)$, $x_j = 0$ for all $j$, and $(n_1, \ldots, n_m)\in S$.
	%
	%It is easy to see that if $S$ is unary then $\mu(S)$ is also unary.
	%Recall that unarity of $\mu(S)$ means that there exists $n$
	%such that it suffices to count exactly till $n$ and above $n$ only modulo $n$ in order to decide whether a vector belongs to $\mu(S)$.
	%Therefore one can build a finite automaton which recognizes $L(\A, \mu(S))$. Indeed, for every state $q$ of $\A$ it remembers
	%whether there is a run over input word reaching $q$.
	%If there is more than one run reaching $q$ it is not necessary to remember these runs, as 
	%by unambiguity of $\A$ no final state is reachable from $q$. If there is exactly one run reaching $q$
	%it is enough to remember for every transition have many times it was used or that it was used more than $n$ times
	%and remember number of visits only modulo $n$.
	%This is a finite information, so indeed $R$ is regular, which finishes proof of this implication.
	%
	\qed
\end{proof}
%

%\begin{remark}
%	\label{rem:complexity:3}
%	In this last step of our reduction,
%	we have reduced regular separability to unary separability
%	of two $\C$-sets of dimension $\leq d^n$,
%	which are obtained by computing the inverse image of an affine linear mapping.
%\end{remark}

%To complete the proof of Lemma~\ref{l:deciddet}
%We now refer to both conditions defining robustness of $\C$.
%First, the sets $V_1$ and $V_2$, being inverse images of sets in $\C$ under affine functions, are themselves in $\C$ and can be computed.
%Second, we can decide unary{\sepsep}separability of $V_1$ and $V_2$.
%Combining these observations with Lemma~\ref{l:tounary} completes the proof of Lemma~\ref{l:deciddet}.

%Therefore in order to finish our algorithm it is enough to decide whether $V_1$ and $V_2$ are unary separable.
%Recall that $V_i = \mu^{-1}(U_i)$ for $i \in \{1, 2\}$ and $\mu: \N^m \to \N^n$ is a affine mapping. Therefore by
%robustness of class $\C$ we get that $V_1, V_2 \in \C$. One more time by robustness we get that unary separability
%of $V_1$ and $V_2$ is decidable, which finishes the proof of Theorem~\ref{thm:reg-sep}.
%

%%% Local Variables:
%%% mode: latex
%%% TeX-master: "main"
%%% End:

% !TEX root = main.tex

\section{Applications}  \label{sec:dec}

We now derive two direct corollaries of Theorem~\ref{thm:reg-sep}. 
In this section by a \emph{projection} we mean a function $\pi_{k, I} : \N^k \to \N^{|I|}$, for $I \subseteq \set{1\ldots k}$, that drops coordinates not in $I$.
%
%Before we need the following
%lemma that show sufficient condition for a class $\C$ to be closed under inverse image 
%of affine map. 
%In the following for $k \in \N$ and $I\subseteq \{1,\ldots, k\}$, we denote by $\pi_I:N^k\to N^{|I|}$
%the projection obtain by restricting to indices in $I$. Formally, if $I = \{i_1\ldots,i_t\}$, then 
%$\pi_I(u) = (u[i_1],\ldots,u[i_t])$.  
%
We start with a simple but useful lemma:
\begin{lemma}\label{lemma:affine}
	If a class $\C \subseteq \bigcup_{d \in \N} \powset{\N^d}$ contains all semilinear sets and  is effectively closed under intersections, 
	projections, and inverse images of projections, then it is effectively closed under inverse images of affine maps.
\end{lemma}

\begin{proof}
	Let $S$ be a set in $\C$ and $f:\N^k\to \N^\ell$ be an affine map 
	defined by $f(u)= Mu+v$ for $M=(m_{i,j})$ a matrix of dimension $\ell \times k$ and $v$ a vector of dimension $\ell$.
	Let $e_j \in \N^k$ be the vector s.t. $e_j[j] = 1$ and $0$ otherwise,
	and let $m_j = (m_{1,j},m_{2,j},\ldots,m_{\ell,j})$ be the (transpose of) the $j$-th column of $M$.
	First remark that the set 
	\[
		E_1 = \{(x,f(x)) \mid x\in\N^k\} \subseteq \N^{k+\ell}
	\]
	is linear with base $(0^k,v)$
	and periods $\set{p_1, \dots, p_k}$,
	where $p_j= (e_j, m_j)\in\N^{k+\ell}$.
	Thus, $E_1 \in \C$.
	Therefore the set $E_2 = E_1 \cap \pi_{k+\ell,I}^{-1}(S)$ is also in $\C$, for  
	$I=\{k+1, \ldots, k+\ell\}$.
	Finally, we conclude since $\pi_{k+\ell, J}(E_2)=f^{-1}(S)$ with $J=\{1, \ldots, k\}$.
	\qed
\end{proof}

%\begin{remark}
%	\label{rem:complexity:inverseimages:semilinear}
%	Let $f(u)= Mu+v$ be an affine map.
%	It's \emph{norm} is the maximal absolute value of any number in $M$ or $v$.
%	%
%	Let $S$ be a semilinear set, whose \emph{norm} is the maximal absolute value of its generators or basis.
%	%
%	When we apply the construction above to obtain $f^{-1}(S)$,
%	we can use the upperbound estimates from \cite[Theorem 6]{ChistikovHaase:Taming:ICALP16}
%	when computing the intersection and we obtain that
%	the inverse image $f^{-1}(S)$ is a semilinear set of norm at most exponential in the dimension,
%	and with at most exponentially many generators (in the dimension).
%\end{remark}

\begin{corollary}\label{corr:pa-decidable}
The regular separability problem is decidable for nondeterministic Parikh automata.
\end{corollary}

\begin{proof}
In order to apply Theorem~\ref{thm:reg-sep} for $\C$ being semilinear sets, we need to know that the class of semilinear sets is robust.
First, Lemma~\ref{lemma:affine} yields effective closure under inverse images of affine maps,
as semilinear sets are effectively closed under boolean combinations, images, and inverse images of projections.
Second, decidability of the unary{\sepsep}separability problem for semilinear sets is a corollary
of the main result in~\cite{DBLP:journals/ipl/ChoffrutG06}.
This theorem states that separability of rational relations in $\Sigma^* \times \N^m$
by recognizable relations is decidable. If we ignore the $\Sigma^*$ component we get
the same result for rational and recognizable relations in $\N^m$, which are exactly semilinear sets and unary sets, respectively.
%A direct polynomial-time procedure for unary{\sepsep}separability of semilinear sets has been recently presented in~\cite{DBLP:journals/corr/ClementeCLP16}.
\qed
\end{proof}

For the second corollary we have to introduce vector addition systems (\vases)
and sections thereof.
%A $d$-dimensional \vass is a triple $V = (s, T, Q)$,
%where $Q$ is a finite set of \emph{states}, $s \in Q \times \N^d$ is the \emph{source} configuration
%and $T \subseteq_\fin Q \times \Z^d \times Q$ is a finite set of \emph{transitions}.
%A \emph{run} $\rho$ of a VASS $V = (s, T, Q)$ is a sequence
%\[
%(q_0, v_0, s_0, q_1, v_1), \ldots, (q_{n-1}, v_{n-1}, s_{n-1}, q_n, v_n)
%\in Q \times \N^d \times \Z^d \times Q \times \N^d
%\]
%such that $(q_0, v_0) = s$ and for all $i \in \{0, \ldots, n-1\}$
%we have $(q_i, s_i, q_{i+1}) \in T$  and $v_i + s_i = v_{i+1}$.
%We write $\target{\rho} =(q_n, v_n)$.
%The \emph{reachability set} of $V$ in state $q$ is
%\[
%\reach_q(V) \ = \ \setof{\target{\rho}}{\rho \text{ is a run of }V}\ \subseteq \ \N^d.
%\]
%The family of all such reachability sets of all \vasses we denote as $\reachvass$.

A $d$-dimensional \emph{vector addition system} (\vas) is a pair $V = (s, T)$, where $s \in \N^d$
is a \emph{source} configuration and $T \subseteq_\fin \Z^d$ is a finite set of \emph{transitions}.
A \emph{run} of a \vas $V = (s, T)$ is a sequence 
$$(v_0, t_0, v_1), (v_1, t_1, v_2), \ldots, (v_{n-1}, t_{n-1}, v_n) \in \N^d \times T \times \N^d$$ 
such that for all $i \in \{0, \ldots, n-1\}$ we have $v_i + t_i = v_{i+1}$ and $v_0 = s$. The
\emph{target} of this run is the configuration $v_n$.
The \emph{reachability set} of a \vas $V$ is the set of targets of all its runs.

In order to ensure robustness, we slightly enlarge the family of \vas reachability sets to \emph{sections} thereof.
The intuition about a section is that we fix values on a subset of coordinates in vectors,
and collect all the values that can occur on the other coordinates.
%
%For a vector $u \in \N^d$ and a subset $I \subseteq \{1, \ldots, d\}$ of coordinates,
%by $\proj I u \in \N^{|I|}$ we denote the \emph{$I$-projection} of $u$, i.e., the vector obtained from $u$ by removing coordinates not belonging to $I$.
For a subset $I \subseteq \{1, \ldots, d\}$,
the projection $\pi_{d, I}$ extends element-wise to sets of vectors $S \subseteq \N^d$, denoted $\pi_{d, I}(S)$.
For a vector $u \in \N^{d-|I|}$, the \emph{section} of $S$ w.r.t. $I$ and $u$ is the set
\begin{align*}
	% \sec{I, u} S :=
	\pi_{d, I}({\{ v \in S \mid \pi_{d, \{1, \ldots, d\} \setminus I}(v) = u \}}) \subseteq \N^{|I|}\enspace.
\end{align*}
We denote by $\Cvas$ the family of all sections of \vasreach sets.

\begin{corollary}\label{corr:pa-vas-decidable}
The regular{\sepsep}separability problem is decidable for nondeterministic $\Cvas$-Parikh automata.
\end{corollary}

\begin{proof}
We apply Theorem~\ref{thm:reg-sep} for $\C = \Cvas$; we thus need to show that class $\Cvas$ is robust.
Decidability of unary{\sepsep}separability of sets from $\Cvas$ is shown in
Theorem 9 in~\cite{DBLP:journals/corr/ClementeCLP16}.
Effective closure of $\C$ under inverse images of affine functions will follow
by Lemma~\ref{lemma:affine} once we prove all its assumptions.

First, $\Cvas$ contains all semilinear sets. Effective closure under intersections is shown in
Proposition 7 in~\cite{DBLP:journals/corr/ClementeCLP16}. Effective closure under inverse images of projections is easy:
extend the \vas with additional coordinates, and allow it to arbitrarily increase these coordinates.

Finally, to see that $\Cvas$ is effectively closed under projections consider a section $S \subseteq \N^d$ of the reachability set of a \vas $V$, and a
subset of coordinates $I \subseteq \{1, \ldots, d\}$.
We construct  a \vas $V'$ which is like $V$, but additionally allows to decrease every coordinate from $\{1, \ldots, d\} \setminus I$. 
Projection $\pi_{d, I}(S)$  of $S$ onto $I$ is a section of the reachability set of $V'$ defined similarly as $S$, but with an additional requirement that
all coordinates from $\{1, \ldots, d\} \setminus I$ have value 0.
\qed
\end{proof}

\bibliographystyle{plain}
\bibliography{citat}

\end{document}